\documentclass[aps,twocolumn,prx,superscriptaddress,longbibliography,reprint]{revtex4-1}

\usepackage{amssymb, amsthm, amsmath}
\usepackage{hyperref}
\usepackage{graphicx}
\usepackage[margin=1in]{geometry}
\hypersetup{breaklinks=true}

\usepackage{algorithm}
\floatname{algorithm}{Protocol}

\usepackage{color}

\newtheorem{theorem}{Theorem}

\newtheorem{definition}{Definition}

\input{Qcircuit}

\begin{document}

\sloppy

\title{A simple protocol for fault tolerant verification of quantum computation}
\author{Alexandru Gheorghiu}
\affiliation{School of Informatics, University of Edinburgh, 10 Crichton Street, Edinburgh EH8 9AB, UK}
\author{Matty J. Hoban}
\affiliation{Clarendon Laboratory, Department of Physics, University of Oxford, Parks Road, OX1 3PJ, UK}
\author{Elham Kashefi}
\affiliation{School of Informatics, University of Edinburgh, 10 Crichton Street, Edinburgh EH8 9AB, UK}
\affiliation{CNRS LIP6, Universit\'{e} Pierre et Marie Curie, Paris, France}

\begin{abstract}
With experimental quantum computing technologies now in their infancy, the search for efficient means of testing the correctness of these quantum computations is becoming more pressing. An approach to the verification of quantum computation within the framework of interactive proofs has been fruitful for addressing this problem. Specifically, an untrusted agent (prover) alleging to perform quantum computations can have his claims verified by another agent (verifier) who only has access to classical computation and a small quantum device for preparing or measuring single qubits. 
However, when this quantum device is prone to errors, verification becomes challenging and often existing protocols address this by adding extra assumptions, such as requiring the noise in the device to be uncorrelated with the noise on the prover's devices. In this paper, we present a simple protocol for verifying quantum computations, in the presence of noisy devices, with no extra assumptions. This protocol is based on \emph{post hoc} techniques for verification, which allow for the prover to know the desired quantum computation and its input. We also perform a simulation of the protocol, for a one-qubit computation, and find the error thresholds when using the qubit repetition code as well as the Steane code.
\end{abstract}

\maketitle

\section{Introduction}
There is now substantial evidence that quantum computers cannot be simulated efficiently by their classical counterparts. Shor's factoring algorithm is one example where an efficient (polynomial time) solution to a problem can be found with a quantum computer, but the best classical algorithm that we know runs in superpolynomial time (in the worst case) \cite{shor}. 
While the inability to be efficiently classically simulated can be of great use in computing, it does raise other problems. In particular, how can we check that the device is producing the correct answer if it is hard to simulate with a classical computer? For the case of factoring, we can just multiply the factors, but for other problems (such as simulating a quantum system) there is no \textit{a priori} classically efficient means of verifying whether a quantum computation was carried out \cite{vazirani}.

To be a bit more explicit, the scenario is that of a \textit{verifier} that can do probabilistic classical computation, and an untrusted \textit{prover} who is limited to universal quantum computation. The verifier wishes to use the prover's quantum computer, but might not be able to trust the prover's output, and so the verifier will perform some sort of verification. In full generality, the prover and the verifier can exchange multiple rounds of classical communication with each other (as the verifier has no quantum capabilities). Both parties' computations have a running time that is, at most, a polynomial in the size of the input to the computation (hence efficient by standard notions in computational complexity). The verifier would like to be able to perform any efficient quantum computation by delegating the task to the untrusted prover. Clearly, if the prover is honest and actually has a quantum computer, they can perform the task for the verifier. On the other hand, if the prover is dishonest then they could lie or deviate in any manner they deem fit. The goal of verification of quantum computation is to catch (with high probability) when a prover is being dishonest and reject his output.

There has been a lot of progress in the development of verification protocols (see \cite{review} for a survey paper). 
It is still an open problem whether verification is possible in this setting, with no extra computational assumptions.
Two main modifications of the classical client, single server setting have been considered: give the verifier some quantum device and allow some quantum communication between the prover and verifier, or introduce multiple non-communicating quantum provers that share entanglement. In the former approach -- the approach we take in this paper -- this quantum device could prepare particular quantum states or make measurements in a particular basis, and this quantum process becomes an integral part of the prover's quantum computation. In this approach, existing protocols typically assume an ideal setting in which the only ``errors'' that can occur are a result of malicious behaviour of the prover.
But realistically, quantum devices are highly susceptible to noise, and so a verifier could introduce noise into a computation implemented by an honest prover. Reducing this noise and controlling it is one of the great challenges in developing scalable quantum computers. The \emph{threshold theorem} shows that as long the error rate per quantum gate is below a constant threshold, it is always possible to perform a \emph{fault tolerant quantum computation} with only a polylogarithmic increase in overhead \cite{threshold}.

We would like to avoid the verifier's quantum device becoming too powerful, even being polylogarithmic in the size of the input to the computation would be too powerful \footnote{For polylogarithmic size (uniform) quantum circuits, there is no known general means of classically simulating them, since a brute force simulation would run in quasipolynomial time.}. Therefore, we want to restrict it to having a quantum register that is constant-size. Even if the errors of the verifier's device can be suppressed, it still needs to be proven that this is not detrimental for the verification of a quantum computation. To wit, we want that a malicious prover does not exploit these errors in order to successfully trick the verifier into accepting incorrect results. 

For protocols in which the verifier is fully classical, fault tolerance is not a concern since one can assume that the provers are performing their quantum operations on top of a quantum error correcting code. Since provers are assumed to have universal quantum computing power, we naturally have to assume that they are capable of fault tolerant quantum computation between themselves. We emphasise that discussions about fault tolerance only make sense in the setting in which the verifier possesses a quantum device.

Thus, we are faced with the following problem:

\textbf{Problem statement:} \emph{Can a verifier with a constant-size and imperfect quantum device verifiably delegate a quantum computation to a single prover?}

We show that this is indeed possible. Furthermore, it is possible even if the verifier's device is an imperfect \emph{single-qubit} measurement device. Our approach is based on that of \textit{post hoc verification} \cite{posthocpublished,posthoc1,posthoc2}, where a prover sends quantum systems to a verifier that should be the ground state of a Hamiltonian. This ground state encodes the desired quantum computation and can be used to ``read off'' the outcome of that computation. If the verifier can indeed certify that this is the ground state, then the computation is verified. In our protocol we encode the qubits of this ground state into a \textit{logical ground state} where each qubit of the original state is encoded into a larger number of physical qubits via a quantum error correction code. This logical state is then the ground state of a logical Hamiltonian described by the quantum computation. In the protocol, the physical qubits in this logical state are then measured one at a time, and appropriate classical corrections are made on the outcomes of these measurements in post-processing if errors are detected. An honest prover's probability of successful computation will be boosted by this error correction, but importantly we can still verify if the logical ground state was indeed prepared by the prover. 

Finally, we consider a simple example of this protocol in the honest prover scenario. That is, using the repetition code and the Steane code, we can simulate and characterise the protocol's behaviour under bit-flip errors and depolarizing noise.

\textbf{Paper outline} --- In Section \ref{sec2} we give some basic complexity theoretic notions to formalize what we mean by verifying efficient quantum computations. We also outline post hoc quantum verification, which is the basis for our approach.
Next, in Section \ref{sec3} we give our protocol for fault tolerant verification of quantum computation, and also prove its correctness; we also describe our simulation of the protocol, with various degrees of noise and outline the obtained results. We then conclude, in Section \ref{sec4}, with some discussions and open problems.

Let us first comment on approaches that have also addressed the aforementioned problem of fault tolerant verification.

\textbf{Related work} --- For protocols in which the verifier has a small quantum device, the question of fault tolerance has been addressed in \cite{gkw, kd, fh, abem}. In \cite{gkw, kd, fh} the authors proposed protocols in which a classical client possessing either a single qubit preparation or measurement device, susceptible to noise, could verifiably delegate quantum computations to a prover. All these protocols are \textit{computationally blind}, meaning that the delegated computation is kept secret from the prover. We will return to this issue in detail in Section \ref{sec4}. Moreover, blindness is \emph{required} for achieving verifiability. However, this requirement of blindness introduces new difficulties when considering fault tolerant computation. To circumvent these difficulties, extra (potentially unrealistic) assumptions were made about the noise, which rule out the possibility of the prover utilising the noise to deceive the verifier. A discussion of the general difficulty in realizing a verifiable, blind, fault tolerant protocol is provided in \cite{abem}.

\section{Preliminaries}\label{sec2}

\subsection{Complexity theory}
Complexity theory classifies computational problems as sets of ``yes/no'' decision problems that are solvable by a particular model of computation, under certain constraints. Decision problems are modeled as sets of binary strings, known as \emph{languages}. The input to the problem is a bit string and the output is yes or no, depending on whether the string belongs to the language or not.
The primary class that is of interest to us, is that of decision problems decidable efficiently by a quantum computer, which is denoted $\mathsf{BQP}$. By ``efficiently'' we always mean in a number of time steps that scales as some polynomial in the size of the input to the problem.
For completeness we reproduce the definition of this class.
\begin{definition}
A language $L \subseteq \{0, 1\}^*$ belongs to $\sf{BQP}$ iff there exist a polynomial $p$, and a uniform quantum circuit family $\{C_{n}\}_{n}$, such that for any $x \in \{0, 1\}^{n}$ the following is true:
\begin{itemize}
\item when $x \in L$, $C_{n}(x)$ accepts with probability at least $a$, and
\item when $x \not\in L$, $C_{n}(x)$ accepts with probability at most $b$,
\end{itemize}
where $a - b \geq 1/p(n)$ and $|C_{n}| \leq p(n)$. 
\end{definition}
If we replace quantum circuits with classical boolean circuits, having access to random bits, we obtain the class $\sf{BPP}$, of problems that can be decided efficiently on a classical computer. We will frequently refer to machines that can solve either $\sf{BPP}$ problems or $\sf{BQP}$ problems as $\sf{BPP}$ machines or $\sf{BQP}$ machines, respectively.

Another class of interest is $\sf{MA}$ which consists of decisions problems for which the ``yes'' instances can be checked by a $\sf{BPP}$ machine, when given access to a polynomial-sized bitstring known as a \emph{proof}. More formally, $\sf{MA}$ is defined as follows:
\begin{definition}
A language $L \subseteq \{0, 1\}^*$ belongs to $\sf{MA}$ iff there exist a polynomial $p$, and a $\sf{BPP}$ machine $\mathcal{V}$ (known as verifier), such that for any $x \in \{0, 1\}^{n}$ the following is true:
\begin{itemize}
\item when $x \in L$, there exists a string $w \in \{0, 1\}^{ \leq p(n)}$ such that $\mathcal{V}(x,w)$ accepts with probability at least $a$, and
\item when $x \not\in L$, for all strings $w \in \{0, 1\}^{ \leq p(n)}$, $\mathcal{V}(x,w)$  accepts with probability at most $b$,
\end{itemize}
where $a - b \geq 1/p(n)$. 
\end{definition}
Essentially, one can view problems in $\sf{MA}$ as those for which a computationally powerful prover can convince a $\mathsf{BPP}$ verifier that the answer is ``yes'', with high probability, by providing a proof string that the verifier can check.
There is a quantum analogue of this known as $\sf{QMA}$, in which the proof string is a quantum state. Specifically:
\begin{definition}
A language $L \subseteq \{0, 1\}^*$ belongs to $\sf{QMA}$ iff there exist a polynomial $p$, and a $\sf{BQP}$ machine $\mathcal{V}$ (known as verifier), such that for any $x \in \{0, 1\}^{n}$ the following is true:
\begin{itemize}
\item when $x \in L$, there exists a quantum state $\ket{\psi}$ having at most $p(n)$-many qubits, such that $\mathcal{V}(x,\ket{\psi})$ accepts with probability at least $a$, and
\item when $x \not\in L$, for all quantum states $\ket{\psi}$ having at most $p(n)$-many qubits, $\mathcal{V}(x,\ket{\psi})$  accepts with probability at most $b$,
\end{itemize}
where $a - b \geq 1/p(n)$. 
\end{definition}

Clearly, $\mathsf{BQP} \subseteq \mathsf{QMA}$, since the $\sf{BQP}$ verifier can simply ignore the proof state from the prover. It is believed that the containment is strict, since, in principle, the prover can produce proof states that cannot be generated by the poly-time quantum verifier. In fact, it was shown in \cite{qmasinglequbit} that the ``quantum overhead'' of the verifier can be reduced to simply performing single-qubit measurements, while maintaining the ability to correctly decide all problems in $\mathsf{QMA}$. This is achieved by instructing the prover to prepare special states that satisfy two properties:
\begin{itemize}
\item Any $\mathsf{BQP}$ computation can be performed through suitable \emph{single-qubit} measurements of these states.
\item Any such state of a given size can be tested through \emph{single-qubit} measurements.
\end{itemize}
If one adds a further condition, namely that these special states can be prepared by a $\mathsf{BQP}$ machine (essentially restricting the prover to $\mathsf{BQP}$), then one obtains a scheme for verifying an arbitrary $\mathsf{BQP}$ computation using only single-qubit measurements. This is precisely what Fitzsimons, Hajdu\v{s}ek and Morimae have done in their post hoc verification protocols \cite{posthoc1,posthoc2,posthocpublished}. One could ask whether this merely allows the verifier to check the ``yes'' instances of $\mathsf{BQP}$ problems, however $\mathsf{BQP}$ is closed under complement\footnote{A complexity class is closed under complement, if for all languages $L$ contained in that class, the complement of $L$, denoted $L^c$ and consisting of all strings not contained in $L$, is also contained in the class.} meaning that the ``no'' instances can also be verified.

\subsection{Post hoc verification}
As mentioned, the core idea of post hoc verification is to have a $\mathsf{BQP}$ prover (or provers) prepare a quantum proof state that the verifier can check using single-qubit measurements.
To explain how the protocol works, we first define a problem known as the \emph{k-local Hamiltonian problem}, which was introduced by Kitaev in \cite{kitaev}.
A $k$-local Hamiltonian, acting on a system of $n$ qubits, is a hermitian operator $H$ that can be expressed as $H = \sum_{i} H_i$, where each $H_i$ is a hermitian operator which acts non-trivially on at most $k$ qubits.
The $k$-local Hamiltonian problem, for which we have taken the definition from \cite{qpcp}, is then the following:
\begin{definition}[The $k$-local Hamiltonian (LH) problem] \ 
\label{def:LH}
  \begin{itemize}
    \item  \textbf{Input:} $H_1,\ldots,H_m$, a set of $m$ Hermitian matrices
      each acting on $k$ qubits  
      out of an $n$-qubit system and
      satisfying $\|H_i\|\le 1$. Each matrix entry is specified by
      $poly(n)$-many bits. Apart from the $H_i$ we are also given
      two real numbers, $a$ and $b$ (again, with polynomially many
      bits of precision) such that $a-b>1/poly(n)$.
            
    \item \textbf{Output:} Is the smallest eigenvalue 
      of $H=H_1+H_2+...+H_m$ smaller than $b$ or are all its
      eigenvalues larger than $a$?
  \end{itemize} 
\end{definition}
Kitaev showed that this problem is complete for the class $\mathsf{QMA}$. In other words, the problem is in $\mathsf{QMA}$ and any problem in $\mathsf{QMA}$ can be reduced to it, in (classical) polynomial time \cite{kitaev}.
The idea is essentially this: for some language $L \in \sf{QMA}$, and given $a$ and $b$, one can construct a $k$-local Hamiltonian such that, whenever $x \in L$, its smallest eigenvalue is less than $b$ and whenever $x \not\in L$, all of its eigenvalues are greater than $a$. 
The proof state, $\ket{\psi}$, when $x \in L$, is the eigenstate of $H$ corresponding to its lowest eigenvalue (or a state that is close, in trace distance, to this state), known as the ground state. The verifier receives the state from the prover and measures one of the local terms $H_i$ (which is an observable) on that state. One can prove that this can be done with a polynomial-sized quantum circuit. This yields an estimate for measuring $H$ itself. Therefore, when $x \in L$ and the prover sends $\ket{\psi}$, with high probability the verifier will obtain the corresponding eigenvalue of $\ket{\psi}$ which will be smaller than $b$.

If the prover is malicious then it would have to convince the verifier to accept when $x \not\in L$. However, when this is the case, all the eigenvalues of $H$ will be larger than $a$ and so, no matter what state the prover sends, when the verifier measures the local term $H_i$ it will, with high probability, obtain a value greater than $b$ and will therefore reject.
We will refer to $a - b$ as the \emph{promise gap} of the local Hamiltonian.

The constant $k$ in the definition of the $k$-local Hamiltonian problem is not arbitrary. In the initial construction of Kitaev, $k$ had to be larger than $5$ for the problem to be $\sf{QMA}$-complete. Subsequent work has shown that even with $k=2$ the problem remains $\sf{QMA}$-complete \cite{2local}.
In fact, $H$ can be a particular type of $2$-local Hamiltonian known as an $XZ$-Hamiltonian and this is the type of Hamiltonian used in the post hoc protocol of Morimae and Fitzsimons \cite{posthoc2, posthocpublished}.

To define an $XZ$-Hamiltonian, we introduce some helpful notation. Consider an $n$-qubit operator $S$, which we shall refer to as $XZ$-term, such that:
\begin{equation}
S = \bigotimes_{j=1}^{n} P_j
\end{equation}
with $P_j \in \{I, X, Z\}$, where $X$ and $Z$ are the Pauli $X$ and $Z$ operators and $I$ is the identity. Denote $w_X(S)$ as the $X$-weight of $S$, representing the total number of $j$'s for which $P_j = X$. Similarly denote $w_Z(S)$ as the $Z$-weight for $S$.
An $XZ$-Hamiltonian is then a $2$-local Hamiltonian of the form:
\begin{equation}
H = \sum_i a_i S_i
\end{equation}
where the $a_i$'s are real numbers and the $S_i$'s are $XZ$-terms having $w_X(S_i) \leq 1$ and $w_Z(S_i) \leq 1$. 
Essentially, as the name suggests, an $XZ$-Hamiltonian is one in which each local term has at most one $X$ operator and one $Z$ operator.

We can now explain the post hoc protocol of \cite{posthoc2}. The protocol relies on the observation that $\sf{BQP} \subseteq \sf{QMA}$. This means that any problem in $\sf{BQP}$ can be viewed as an instance of the $2$-local Hamiltonian problem.
Therefore, for any language $L \in \sf{BQP}$ there exists an $XZ$-Hamiltonian, $H$, and a polynomial-time quantum verifier which can measure a local term of $H$ on the quantum witness sent by the prover and decide the problem correctly, with high probability. But since the local terms of an $XZ$-Hamiltonian have at most one $X$ operator and one $Z$ operator, the verifier will essentially have to measure at most two qubits, one in the standard basis $(\ket{0}, \ket{1})$ and the other in the Hadamard basis $(\ket{+}, \ket{-})$. 

We now restrict attention to $L \in \mathsf{BQP}$.
As mentioned, when $x \in L$, the proof state that the prover should send to the verifier should be close to the ground state of the $XZ$-Hamiltonian.
When $L \in \mathsf{BQP}$, the Hamiltonian can be chosen so that the ground state is close to a particular type of state known as a \emph{Feynman-Kitaev clock state} (also known as \emph{history state}), which can be prepared by a $\mathsf{BQP}$ prover \cite{kitaev}. 
To describe this state, consider a quantum circuit $\mathcal{C} = U_{T} U_{T-1} ... U_1$, with classical input $\ket{x}$, where $T = poly(|x|)$, for testing whether $x \in L$. Denoting $U_0 = I$, the Feynman-Kitaev state associated to $\mathcal{C}$ and $\ket{x}$ is the following:
\begin{equation}
\ket{\psi} = \frac{1}{\sqrt{T + 1}} \sum\limits_{t=0}^{T} U_{t} U_{t-1} ... U_0 \ket{x} \ket{1^{t} 0^{T - t}}
\end{equation}
There exists an $XZ$-Hamiltonian, $H$, such that when $x \in L$, we have that $\bra{\psi} H \ket{\psi} < a$, and when $x \not\in L$ we have that for any $\ket{\phi}$, $\bra{\phi} H \ket{\phi} > b$, for some $a, b$ such that $a - b > 1/poly(|x|)$.
The exact form of $H$ is not important for understanding the protocol. What is important is that for any $L \in \mathsf{BQP}$, the verifier can efficiently compute the description of $H$.

The post hoc protocol then works as follows:
\begin{enumerate}
\item The verifier computes the terms $a_i$ of the $XZ$-Hamiltonian, $H = \sum_i a_i S_i $, corresponding to $L$ and input $x$. They then send the description of $H$ and $x$ to the prover.
\item The prover responds by preparing the ground state of $H$ (the Feynman-Kitaev state, described above), denoted $\ket{\psi}$, and sends it to the verifier. This constitutes the quantum proof state for the statement that $x \in L$ (if $x \not\in L$, the same procedure is performed for the complement of $L$, denoted $L^c$, which is also in $\mathsf{BQP}$).
\item The verifier chooses one of the $XZ$-terms $S_i$, according to the normalized probability distribution $\{|a_i|\}_i$, and measures it on $\ket{\psi}$. They accept on outcome $-sgn(a_i)$ of the measurement.
\end{enumerate}
The protocol is correct, in that when the prover aims to convince the verifier that $x \in L$ (or $x \in L^c$, respectively) and sends the correct state (for the Hamiltonian corresponding to $L$ or to $L^c$, respectively), the verifier will accept with probability:
\begin{equation}
p_{acc} \geq \frac{1}{2} \left( 1  - \frac{b}{\sum_i|a_i|} \right)
\end{equation}
Additionally, the protocol is sound in that when the prover aims to convince the verifier that $x \not\in L$ (or $x \not\in L^c$, respectively), irrespective of the state that the prover sends, the verifier will accept with probability:
\begin{equation}
p_{acc} \leq \frac{1}{2} \left( 1  - \frac{a}{\sum_i|a_i|} \right)
\end{equation}
Since $\sum_i |a_i|$ is a constant and $a -b > 1/poly(|x|)$, the gap between the two probabilities is inverse polynomial in the size of the input.

\section{Fault tolerant verification of quantum computation}\label{sec3}
The protocol described above works assuming an ideal setting in which the quantum devices of both the prover and the verifier are perfect. Of course, this is an unrealistic assumption since any implementation of the protocol will be subject to noise stemming from an imperfect isolation of the quantum systems from the environment, and the use of faulty devices. 
It is straightforward to show that a constant rate of noise on these devices will lead to the failure of the protocol for sufficiently large computations.
This is because the gap between acceptance and rejection, defined by $a - b$, is inverse polynomial in the size of the input. As a result of noisy devices, the acceptance threshold is shifted to $a - c$, and the rejection threshold is shifted to $b + c$, where $c$ is some positive constant that depends on the noise rate of the devices. We can see that as long as $c < (a - b) / 2$, the verifier can still distinguish reliably between acceptance and rejection. However, it is clear that for a sufficiently long input, we will have that $c \geq (a - b) / 2$. At this point, the protocol no longer satisfies the correctness nor the soundness criteria.
In fact, this is common to all other verification protocols in the single-prover setting \cite{review}.
To address this issue we now give a fault tolerant version of the post hoc protocol that works in the presence of quantum devices subject to local noise having a constant error-rate.

\subsection{The fault tolerant protocol}
Our construction is simple: we ask the prover to encode the history state in a CSS (Calderbank-Shor-Steane) error-correcting code \cite{nc} and send it to the verifier. The verifier will then perform a \emph{transversal} measurement of the $X$ and $Z$ operators.
Transversality results in the logical operators being expressed as tensor products of physical $X$ and $Z$ operators, i.e.:
\begin{equation}
\tilde{X} = \bigotimes_{i=1}^m X_i \quad \quad \tilde{Z} = \bigotimes_{i=1}^m Z_i 
\end{equation}
where $\tilde{X}$ and $\tilde{Z}$ are the logical (or encoded) $X$ and $Z$ operators.
In effect, the original Hamiltonian is replaced with an encoded Hamiltonian by substituting each $XZ$-term with its corresponding logical form.

CSS codes are transversal and this ensures that the verifier needs to perform only single-qubit measurements. We also require an additional property, that is possessed by CSS codes, namely that the outcomes for the transversal measurements (of the $X$ and $Z$ operators) are encoded in a classical error-correcting code. This is because the verifier will not perform any quantum correction on the state sent by the prover. Instead, this state will be measured and the measurement outcomes are classically post-processed.

To clarify, consider the following simple example. Assume that the CSS code is a repetition code in which $\tilde{\ket{0}} = \ket{0}^{\otimes m}$ and $\tilde{\ket{1}} = \ket{1}^{\otimes m}$, for some odd $m > 1$. This code can correct $\lfloor \frac{m}{2} \rfloor$ bit-flip errors. If the verifier wishes to measure the $\tilde{Z}$ observable on an encoded state, they will instead measure $\bigotimes_{i=1}^m Z_i$. The $m$-bit outcome corresponds to the outcome of $\tilde{Z}$ encoded in a classical repetition code. Thus, the verifier will simply take the majority bit as the outcome of $\tilde{Z}$.

For our protocol, the verifier will measure a local term of the encoded Hamiltonian, in a transversal way, and perform the classical post-processing of the results in order to extract the corrected measurement outcome. With this corrected outcome, the acceptance condition is the same as in the ``unencoded'' case (i.e. if the outcome for the measurement of term $\tilde{S}_i$ is $-sgn(a_i)$).

To guarantee that this construction works, we show the following:
\begin{itemize}
\item[\textbf{(1)}] The encoded Hamiltonian preserves the $a - b$ promise gap of the original Hamiltonian. This is equivalent to showing that the encoded ground state of the original Hamiltonian is a ground state of the encoded Hamiltonian having the same energy.
\item[\textbf{(2)}] A polylogarithmic number of concatenations of the CSS code is sufficient to maintain an inverse polynomial acceptance-rejection gap in the presence of noise.
\end{itemize}
Having these properties guarantees that the fault tolerant post hoc protocol is both correct and sound, even in the presence of noisy devices. In other words:
\begin{theorem} \label{thm:main}
The post hoc protocol of Morimae and Fitzsimons can be made fault tolerant by encoding the $XZ$-Hamiltonian of the protocol in a CSS code and having the verifier perform the $X$ and $Z$ measurements in a transversal fashion.
\end{theorem}
\begin{proof}
Let $\tilde{X}$ and $\tilde{Z}$ be the logical $X$ and $Z$ operators in the chosen CSS code. We have that $\{ \tilde{X}, \tilde{Z} \} = 0$ and we will assume that these operators act on $m > 0$ qubits. Since these are operators for an error correcting code, there exists an encoding unitary, denoted $E$, such that:
\begin{equation} \label{eqn:enc1}
E (X \otimes I^{\otimes m - 1}) E^{\dagger} = \tilde{X}
\end{equation}
\begin{equation} \label{eqn:enc2}
E (Z \otimes I^{\otimes m - 1}) E^{\dagger} = \tilde{Z}
\end{equation}
Now let $H = \sum_i a_i S_i$ be an $XZ$-Hamiltonian acting on $n > 0$ qubits, and let $H' = H \otimes I^{n(m-1)}$. Clearly, $H$ and $H'$ have the same eigenvalues. But note that using Equations~\ref{eqn:enc1} and~\ref{eqn:enc2} we have that:
\begin{equation}
E^{\otimes n} S_i \otimes I^{n(m-1)} E^{\otimes n} = \tilde{S_i}
\end{equation}
where $\tilde{S_i}$ is obtained by replacing $X$, $Z$ and $I$ by $\tilde{X}$, $\tilde{Z}$ and $I^{\otimes m}$, respectively. This then implies that:
\begin{equation}
E^{\otimes n} H' E^{\otimes n} = \tilde{H}
\end{equation}
where $\tilde{H} = \sum_i a_i \tilde{S_i}$ is the encoded $XZ$-Hamiltonian. Thus, since $\tilde{H}$ and $H'$ are unitarily related, they will also have the same eigenvalues. Moreover, if $\ket{\tilde{\psi}} = E^{\otimes n} \ket{\psi} \ket{anc}$ is the encoded version of some $n$-qubit state $\ket{\psi}$, for a suitably chosen ancilla state $\ket{anc}$, it is clear that for any such $\ket{\psi}$ we have that:
\begin{equation}
\bra{\psi} H \ket{\psi} = \bra{\tilde{\psi}} \tilde{H} \ket{\tilde{\psi}}
\end{equation}
Therefore, if $\ket{\psi}$ is a ground state of $H$, $\ket{\tilde{\psi}}$ will be a ground state of $\tilde{H}$.

This proves property \textbf{(1)}, since it shows that the encoded Hamiltonian will have the same promise gap as the original Hamiltonian.

To prove property \textbf{(2)}, we first need to describe what we mean by noisy measurements. The verifier makes $X$ and $Z$ measurements, but with probability $\epsilon_m$ there is an error in the measurement. The probability of error is independent between uses of the measurement devices, i.e. there are no correlated errors\footnote{Given that we use CSS codes, our construction can tolerate correlated errors as well, provided that these are correctable by the chosen CSS code.}. To be a bit more precise, for ideal measurement operator $M_{x}$ for outcome $x$, we apply a unital map $\mathcal{E}$ to $M_{x}$, where with probability $1-\epsilon_m$, $M_{x}$ is unchanged, and with probability $\epsilon_m$, $M_{x}$ is changed to something else. Alternatively, if we measure an $n$-qubit state $\rho$ one qubit at a time, the noisy measurement is equivalent to transforming $\rho$ to $(\mathcal{E}^{\dagger})^{\otimes n}(\rho)$, and then making an ideal measurement on each qubit individually, where $\mathcal{E}^{\dagger}$ is the channel that is dual to $\mathcal{E}$.

This error model of the measurement device is exactly how errors are traditionally modelled in quantum computation, where they are identically and independently distributed on the qubits. So if each qubit in the Hamiltonian is encoded in a block of qubits, then due to the error-correcting code, the probability of obtaining an incorrect outcome (after classical post-processing) has been suppressed from $\epsilon_m$ on the original qubit to at most $\alpha\epsilon_m^{2}$ on the whole block, for some constant $\alpha$ (determined by the code). 
Here we have implicitly used the fact that the measurement outcome for the logical qubit in one block is obtained through classical error correction (post-processing) of the outcomes of measuring the block qubits.
Concatenating $k$ times then results in probability $\alpha^{(2^{k}-1)}\epsilon_{m}^{2^{k}}$ of there being an error upon measuring an encoded qubit.

The verifier will make two logical qubit measurements, so to achieve a final error rate $\eta$, we must have the error for each logical qubit after $k$ concatenations be $\alpha^{(2^{k}-1)}\epsilon_{m}^{2^{k}}\leq\frac{\eta}{2}$. Provided that $\alpha$ is below the threshold probability $p_{th}=\alpha^{-1}$ of the code, then if each block consists of $b$ qubits with $k$ levels of concatenation, for each qubit we have 
\begin{equation}
b^{k}=\left(\frac{\textrm{log}(2/\alpha\eta)}{\textrm{log}(1/\alpha\epsilon_{m})}\right)^{\textrm{log}b},
\end{equation}
which is $O(\textrm{polylog}(\frac{2}{\eta}))$. So if the total number of qubits in the ground state of the original Hamiltonian is $n$, after $k$ levels of encoding in blocks of size $b$, the total number of qubits in the encoded ground state is $O(n \; \textrm{polylog}(\frac{2}{\eta}))$.

If the probability of acceptance (rejecting) in the original protocol (without noisy measurements) is $p_{acc}$ ($p_{rej}$) and we have that $p_{acc}-p_{rej}\leq\frac{1}{\textrm{poly(n)}}$. Now with noisy measurements, we have that the new probability of acceptance (with error correction) is $\tilde{p}_{acc}\geq p_{acc}-\eta$ and $\tilde{p}_{rej}\leq p_{rej}+\eta$. Therefore, to maintain a polynomial gap between acceptance and rejection we must have that $\eta$ is sufficiently smaller than an inverse polynomial, which only incurs a polylogarithmic overhead. Note that only a polynomial overhead is required if we wish for $\eta$ to be exponentially small.
\end{proof}

The idea of encoding the proof state in an error-correcting code while maintaining a single-qubit measurement device for the verifier has also been considered, in the context of general $\sf{QMA}$ problems, in \cite{qmaft}. In that case, however, the proof state is a graph state that is used by the verifier to perform a fault tolerant measurement-based quantum computation. The verifier is also required to test that this state corresponds to the correct graph state and this is achieved through a stabilizer test.

In our case, by restricting to $\sf{BQP}$ computations, we simply require the verifier to measure the history state associated to the quantum computation. By showing that the encoded Hamiltonian has the same promise gap as the original Hamiltonian it is therefore sufficient to request that the prover encode the history state in a CSS code.

\subsection{Example}
Let us consider a toy example of our protocol in the case of an honest prover, for which we will give numerical results when using the repetition code and the Steane code, respectively.
To start with, we should consider a quantum computation for which we want to construct a history state. Given that the Steane code will encode one logical qubit as $7$ physical qubits, this computation needs to be small enough so that we are able to perform multiple runs of the protocol, in a reasonable amount of time.
For this reason, we will choose the following one-qubit computation:
\begin{figure}[htbp!]
\[
\Qcircuit @C=2em @R=1.4em {
   & \lstick{\ket{x}} & \gate{X} & \gate{D(\pi/8)} & \qw \\
}
\]
\caption{Example computation.}
\label{fig:circuit}
\end{figure}
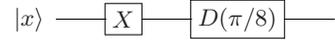

where:
\begin{equation}
D(\phi) = cos(\phi) Z + sin(\phi) X
\end{equation}
Note that $D(\phi)$ is universal for single-qubit quantum computations\footnote{Additionally, $\{ CNOT, D(\phi)\}$ is universal for general quantum computations.}.
The computation has two time steps, hence $T=2$. Consider the case $x=0$. The input state starts out as $\ket{0}$, it is then flipped to $\ket{1}$ and upon application of the $D(\pi/8)$ gate it becomes $sin(\pi/8) \ket{0} - cos(\pi/8) \ket{1}$.
If we designate output $\ket{1}$ as acceptance, then this circuit will accept $x = 0$ with probability $cos(\pi/8)^2$.
The history state, for $x=0$, will be:
\begin{multline*}
\ket{\psi_{x=0}} = \frac{1}{\sqrt{3}} ( \ket{0}\ket{00} + \ket{1}\ket{10} + \\
                 (sin(\pi/8) \ket{0} - cos(\pi/8) \ket{1}) \ket{11} ) \\
\end{multline*}
where we have separated the computation register from the clock register.
For the $x=1$ case, the history state will be:
\begin{multline*}
\ket{\psi_{x=1}} = \frac{1}{\sqrt{3}} ( \ket{1}\ket{00} + \ket{0}\ket{10} + \\
                 (cos(\pi/8) \ket{0} + sin(\pi/8) \ket{1}) \ket{11} ) \\
\end{multline*}

We now need to consider an $XZ$-Hamiltonian such that the ground state is close to $\ket{\psi_{x=0}}$. Since the $2$-local construction is fairly involved and we are only interested in a simple example, we will instead consider a $3$-local Hamiltonian. This, of course, does not change the protocol in any way and the verifier will still perform single-qubit $X$ and $Z$ measurements. Following the works of \cite{kemperegev, 2local}, the Hamiltonian will have the following form:
\begin{equation*}
H = H_{in} + H_{clock} + H_{prop} + H_{out}
\end{equation*}
where:
\begin{itemize}
\item $H_{in}$ penalizes terms in which the input is not of the correct form, at the start of the computation ($T=0$).
\item $H_{clock}$ penalizes terms in which the clock register is not of the correct form, throughout the computation.
\item $H_{prop}$ penalizes terms that do not correspond to the chosen computation.
\item $H_{out}$ penalizes terms for which the output of the computation register is not $\ket{1}$ (i.e. non-accepting computations).
\end{itemize}
In our case, we have:
\begin{equation*}
H_{in} = (I - \ket{x}\bra{x}) \otimes \ket{0}\bra{0} \otimes I
\end{equation*}
\begin{equation*}
H_{clock} = I \otimes \ket{01}\bra{01}
\end{equation*}
\begin{equation*}
H_{prop} = H_{prop_1} + H_{prop_2}
\end{equation*}
where:
\begin{equation*}
H_{prop_1} = \frac{1}{2}( I \otimes \ket{0}\bra{0} \otimes I - X \otimes X \otimes I + I \otimes \ket{10}\bra{10})
\end{equation*}
\begin{equation*}
H_{prop_2} = \frac{1}{2}( I \otimes I \otimes \ket{1}\bra{1} - D(\pi/8) \otimes I \otimes X + I \otimes \ket{10}\bra{10})
\end{equation*}
and finally:
\begin{equation*}
H_{out} = \ket{0}\bra{0} \otimes I \otimes \ket{1}\bra{1}
\end{equation*}
It should be noted that $\ket{\psi_{x}}$ is the ground state of $H_{in} + H_{clock} + H_{prop}$, but not the ground state of $H$. It is the $H_{out}$ term that singles out $\ket{\psi_{x=0}}$ and makes the ground state of $H$ be close, in trace distance, to the history state for the $x=0$ case. This is because in that case, the output of the computation will be $\ket{1}$, with high probability.

We now write $H$ in $XZ$ form:
\begin{multline} \label{eqn:xzham}
H = \frac{7}{4} III + \frac{1}{4}  (1 - (-1)^x) ZII - \frac{1}{4} (-1)^x ZZI \\
-\frac{1}{4}IZZ - \frac{1}{2} XXI - \frac{1}{2} XXZ - \frac{1}{2} sin(\pi/8) XIX \\
+ \frac{1}{2} sin(\pi/8) XZX - \frac{1}{2} cos(\pi/8) ZIX \\
 + \frac{1}{2} cos(\pi/8) ZZX - \frac{1}{4} ZIZ
\end{multline}

The protocol proceeds as follows. The verifier will inform the prover that they wish to perform the computation from Figure~\ref{fig:circuit}, for input $x=0$. The prover reports that the computation accepts (with high probability) and prepares the history state $\ket{\psi_{x=0}}$, encoded in a CSS code. This state is sent qubit by qubit to the verifier. The verifier, will choose one of the terms from Equation~\ref{eqn:xzham}, with its corresponding probability, and perform the transversal measurement of the state. For instance, the term $XZX$ will be chosen with probability $\frac{1}{2K} sin(\pi/8)$, where $K = \sum_i |a_i| \approx 4.8$. The verifier measures the $X$ and $Z$ operators, performs classical post-processing on their results and combines them so as to recover the outcome of measuring $XZX$. She accepts on outcome $-1$ for this measurement, since $\frac{1}{2}sin(\pi/8)$ is positive.

For the $x=1$ case, the situation is similar. In this case, the prover will inform the verifier that the computation rejects (with high probability) and so the verifier will change the $H_{out}$ term of the Hamiltonian to:
\begin{equation}
H_{out} = \ket{1}\bra{1} \otimes I \otimes \ket{1}\bra{1}
\end{equation}
and otherwise proceed as in the $x=1$ case.

\subsection{Numerical results}
To simulate the above protocol, we considered two error-correcting codes: the repetition code and the Steane code. In both instances, we wanted to compare how the verifier's probability of acceptance changes as we increase the amount of noise applied to the history state. Before showing the results, we should first ask: what is the probability of acceptance, for $x=0$, when there is no noise in the system? One can show that:
\begin{equation}
p_{acc} = \frac{1}{2} \left( 1 - \frac{\bra{\psi_{x=0}} H \ket{\psi_{x=0}}}{\sum_i |a_i|} \right)
\end{equation}
and in our case $\bra{\psi_{x=0}} H \ket{\psi_{x=0}} \approx 0.0488$. We therefore find that $p_{acc} \approx 0.4949$.

The first case we considered is the repetition code, with $3$ physical qubits per logical state. This code can only correct for $X$ errors. We therefore considered the noise channel:
\begin{equation*}
\mathcal{F}(\rho) = (1 - p) \rho + p X \rho X
\end{equation*}
acting independently on each individual qubit.
The results are shown in Figure~\ref{fig:rep3}.

\begin{figure}[htbp!]
\centering
\includegraphics[scale=0.57]{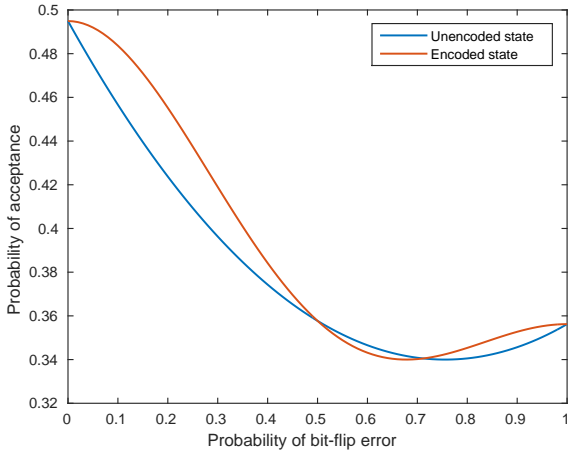}
\caption{Comparison between encoded and unencoded states for the $3$-qubit repetition code.}
\label{fig:rep3}
\end{figure}

As we can see, the point where the encoded state yields the same acceptance probability as the unencoded state is $p = 0.5$. The acceptance probabilities for the unencoded state were determined by applying the channel $\mathcal{F}$ to each qubit in $\ket{\psi_{x=0}}$, resulting in a state $\rho$, and then computing:
\begin{equation} \label{eqn:pacc}
p_{acc} = \frac{1}{2} - \frac{Tr(H \rho)}{2\sum_i |a_i|}
\end{equation}
The same is true for the encoded state, except that logical $Z$ operators are replaced with:
\begin{equation*}
Z_M = M_0 - M_1
\end{equation*}
where:
\begin{equation*}
M_0 = \ket{000}\bra{000} + \ket{001}\bra{001}
 + \ket{010}\bra{010} + \ket{100}\bra{100}
\end{equation*}
\begin{equation*}
M_1 = \ket{111}\bra{111} + \ket{110}\bra{110}
+ \ket{101}\bra{101} + \ket{011}\bra{011}
\end{equation*}
Essentially, the $+1$ eigenspace of $Z_M$ is spanned by states containing a majority of $\ket{0}$ and the $-1$ eigenspace is spanned by states containing a majority of $\ket{1}$. Measuring $Z_M$ is the same as performing a transversal $Z$ measurement and taking the majority outcome.

If we increase the size of the encoded state to $5$ qubits, we obtain the results from Figure~\ref{fig:rep5}.
\begin{figure}[htbp!]
\centering
\includegraphics[scale=0.57]{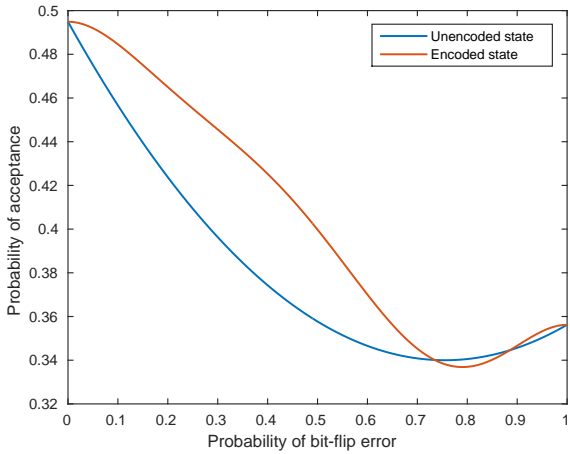}
\caption{Comparison between encoded and unencoded states for the $5$-qubit repetition code.}
\label{fig:rep5}
\end{figure}
As expected, the noise threshold increases and is around $p \approx 0.72$.

We now consider the Steane code, which can detect and correct for arbitrary errors on a single qubit, while encoding one logical state in $7$ physical qubits. This means that the encoded state will comprise of $21$ qubits. 
For this case, we will assume that each qubit is subject to depolarizing noise, characterised by the channel:
\begin{equation*}
\mathcal{D}(\rho) = (1 - 3p/4) \rho + p/4 (X \rho X + Y \rho Y + Z \rho Z)
\end{equation*}

\begin{figure}[htbp!]
\centering
\includegraphics[scale=0.57]{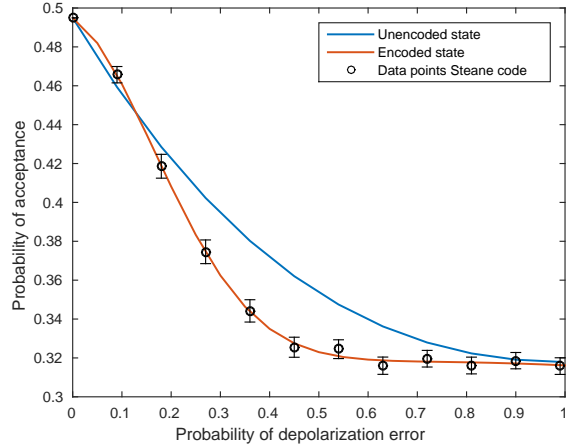}
\caption{Comparison between encoded and unencoded states for Steane's code.}
\label{fig:steane1}
\end{figure}

Due to the large number of entries for the density matrix of the encoded state, we were unable to directly apply the channel $\mathcal{D}$. Instead, for each qubit in $\ket{\tilde{\psi}_{x=0}}$, we chose to either leave it unchanged, with probability $(1 - 3p/4)$ or, with probability $p/4$, apply either $X$, $Y$ or $Z$.
This process is repeated multiple times, and in each case the probability of acceptance is computed using Equation~\ref{eqn:pacc}. The overall probability of acceptance is then estimated by taking the average over all of these runs.
The results are shown in Figure~\ref{fig:steane1}.

\begin{figure}[htbp!]
\centering
\includegraphics[scale=0.57]{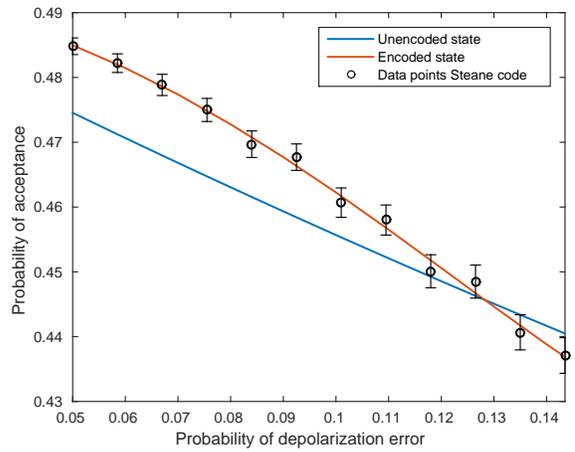}
\caption{Threshold for the Steane code.}
\label{fig:steane2}
\end{figure}

We considered $12$ data points, spread equally in the interval $[0, 1]$, and for each we performed $1000$ repetitions of applying noise in order to estimate $p_{acc}$.
The error bars represent confidence intervals for the computed values, assuming a confidence of $95\%$. Additionally, the orange curve represents the best fit interpolation of the given samples, when assuming a Gaussian model. As we can see, the threshold point appears to be between $0.1$ and $0.2$. 
By considering $12$ samples in the range between $0.05$ and $0.15$, and $4000$ repetitions per sample, in Figure~\ref{fig:steane2}, we find that the threshold point is between $0.12$ and $0.13$.

The simulations were performed in MATLAB, on the Eddie Mark $3$ cluster of The University of Edinburgh. The code for our simulations is available on Github \cite{github}.

\section{Conclusions}\label{sec4}
We have given a simple construction for a fault tolerant quantum verification protocol. In a nutshell, the construction involves taking the original post hoc verification protocol of Morimae and Fitzsimons and encoding it in a CSS error-correcting code. Since the original protocol was not blind, neither is its fault tolerant counterpart. A protocol being blind means that the delegated computation is kept secret from the prover, and they only learn at most the size the computation. A major open problem that remains to be addressed is whether one can achieve fault tolerant verification of blind quantum computation without resorting to additional assumptions, as in \cite{gkw, kd, fh}. Specifically, the protocols from \cite{gkw, kd, fh} assumed (either implicitly or explicitly) that the noise on the verifier's device is independent of the secret parameters that are used to achieve blindness. Additionally, the noise, on that device, should be uncorrelated with the prover's private system.

Following the discussion in \cite{abem}, the authors stress that, so far, there is no protocol that simultaneously achieves all of the following properties:

\begin{itemize}
\item[\textbf{(1)}] The verifier has a preparation or measurement device whose size is at most polylogarithmic in the size of the delegated quantum computation.
\item[\textbf{(2)}] The noise rate for each quantum operation is below some constant threshold. Additionally, the noise on the verifier's device can depend on whatever operations the verifier performs and can be correlated with the prover's quantum system.
\item[\textbf{(3)}] The protocol is unconditionally blind. In other words, throughout the interaction with the verifier, the prover only learns the size of the delegated quantum computation.
\end{itemize}

As mentioned, previous approaches achieved conditions $1$ and $3$ but not $2$. The protocol we proposed achieves conditions $1$ (with a constant size device) and $2$ but not $3$.

Recently, a protocol has been proposed in which a classical client can delegate and verify the computations performed by a quantum server \cite{urmila}. This protocol, however, relies on certain computational assumptions about whether a quantum computer can solve a particular problem. Therefore, the verifier would not need to worry about introducing errors into the prover's quantum computation, as was the concern in our work, but this comes at the cost of making these computational assumptions. Interestingly, the protocol in \cite{urmila} also uses post hoc verification as a primitive, except now the prover measures the qubits in the history state and relays the outcomes to the verifier. The preparation of the history state is slightly more complex than in our case since it uses cryptographic one-way functions which introduce some overhead.

Returning to our results, the simulations are encouraging. Given that the obtained thresholds are higher than the error rates observed in current experimental implementations \cite{superconducting, takita2017experimental, google}, a demonstration of the protocol in the near future is likely.
The major obstacle to such a demonstration would be the production of these highly entangled history states. The use of CSS codes, however, means that one can encode these states in codes having even higher noise thresholds than the Steane code, such as surface codes \cite{surfacecodes}.

\textbf{Acknowledgments.}
We acknowledge useful correspondence with Thomas Vidick.
This work was supported by EPSRC grants EP/N003829/1 and EP/M013243/1. MJH also acknowledges funding from the EPSRC grant Building Large Quantum States out of Light (EP/K034480/1).

\Urlmuskip=0mu plus 1mu\relax
\bibliographystyle{unsrt}
\bibliography{refs}

\end{document}